\documentclass[runningheads]{llncs}
\usepackage[colorlinks,citecolor=green!60!black,linkcolor=red!60!black]{hyperref}
\usepackage{color}
\usepackage{multirow, array}
\usepackage{tabularx}
\usepackage{amsmath}
\usepackage{amssymb,amsfonts,textcomp}
\usepackage{latexsym}
\usepackage{graphicx} 
\usepackage{microtype}
\usepackage{float}
\usepackage{rotating}
\usepackage{url}
\usepackage{comment}
\usepackage{paralist}
\usepackage{booktabs}
\usepackage{bm}
\usepackage[]{units}
\usepackage[noend, ruled, noline]{algorithm2e}
\usepackage{tikz}
\usepackage[sort&compress,comma,numbers,sectionbib]{natbib}
\usepackage{mathtools}
\usepackage{cleveref}
\usepackage{rotating,tcolorbox,complexity}
\usepackage{thmtools}
\usepackage{thm-restate}

\usepackage{boxedminipage}
\usepackage{xspace}

\allowdisplaybreaks

\newcommand{\ie}{i.e.,\xspace}

\newcommand{\problemdef}[4]{
	\begin{center}
		\begin{minipage}{0.95\textwidth}
			\normalsize\textsc{#2} \smallskip \\
			\begin{tabularx}{\textwidth}{@{}l@{\hspace{3pt}}X}
				\normalsize\textbf{Input:} & \normalsize#3 \\
				\normalsize\textbf{#1:}    & \normalsize#4
			\end{tabularx}
		\end{minipage}
	\end{center}
}
\newcommand{\dprob}[4][Question]{\problemdef{#1}{#2}{#3}{#4}}
\DeclareMathOperator{\vote}{vote}

\newcommand{\Mp}{M^\prime}

\newcounter{captionedequationset} 
\newdimen\captionlength
\newcommand{\captionedequationset}[1]{
   ~\refstepcounter{captionedequationset}%
    \setlength{\captionlength}{\widthof{#1}} 
    \addtolength{\captionlength}{\widthof{Equation set~\thecaptionedequationset: }}
    \ifthenelse{\lengthtest{\captionlength < \linewidth }} 
    {\begin{center}
            Preference structure~\thecaptionedequationset: #1
        \end{center}} 
    { \begin{flushleft} 
        Preference structure~\thecaptionedequationset: #1 
        \end{flushleft}}}

\usepackage[textsize=scriptsize,textwidth=3cm,shadow]{todonotes}

\begin{document}
\title{Three-Dimensional Popular Matching\\ with Cyclic Preferences}
\author{
\'Agnes Cseh\inst{1,2}%
\and
Jannik Peters\inst{3}%
}
\authorrunning{\'A. Cseh \and J. Peters}
\institute{Hasso-Plattner-Institute, University of Potsdam, Germany \and
Institute of Economics, Centre for Economic and Regional Studies, Hungary\\
\email{agnes.cseh@hpi.de}\\
\and TU Berlin, Germany \\
\email{jannik.peters@tu-berlin.de}
}
\maketitle 
\begin{abstract} 
Two actively researched problem settings in matchings under preferences are popular matchings and the three-dimensional stable matching problem with cyclic preferences. In this paper, we apply the optimality notion of the first topic to the input characteristics of the second one. We investigate the connection between stability, popularity, and their strict variants, strong stability and strong popularity in three-dimensional instances with cyclic preferences. Furthermore, we also derive results on the complexity of these problems when the preferences are derived from master lists.
\end{abstract}

\section{Introduction}
Partitioning agents into desirable groups is one of the core problems of algorithmic game theory. However, the lines between tractability and intractability are often very thin; introducing ties, incomplete lists or slight variations to the preference or group structures can make a previously tractable problem intractable. In this work, we aim to further draw this line by studying popularity in three-dimensional matching instances equipped with cyclic preferences. 

\subsection{Problem Setting}
In a \textit{three-dimensional (3D) matching instance}, we are given three sets of agents $A,B,$ and $C$, representing for example users, data sources, and servers~\cite{CJ13} or as it is commonly referred to in the literature~\cite{Man13,NH91}, men, women, and dogs. Each agent in $A,B,$ and $C$ declares a subset of the agents in $B,C,$ and $A$, respectively, acceptable. A matching $M$ consists of $(a,b,c) \in A \times B \times C$ triples such that $a$ finds $b$ acceptable, $b$ finds $c$ acceptable, and finally, $c$ finds $a$ acceptable; furthermore, each agent appears in at most one triple in~$M$.

In the problem variant we study, each agent possesses a strictly ordered preference list. \textit{Cyclic preferences} mean that agents in $A$ have preferences over the acceptable agents in $B$, agents in $B$ have preferences over the acceptable agents in $C$, and finally, agents in $C$ have preferences over the acceptable agents in~$A$. The standard problem is to decide whether such an instance admits a stable matching. Two intuitive stability notions have been investigated in the literature: a \textit{weakly stable matching} does not admit a triple so that all three agents would improve, while according to \textit{strong stability}, a triple already blocks if at least one of its agents improves, and the others in the triple remain equally satisfied.

The optimality criterion we study in this paper is \textit{popularity}, which is a well-studied concept in the context of two-sided matching markets. 
Given two matchings~$M$ and $M'$, matching $M$ is more popular than $M'$ if the number of agents preferring $M$ to $M'$ is larger than the number of agents preferring $M'$ to $M$. A matching~$M$ is called \emph{popular} if there is no matching~$M'$ that is more popular than~$M$. Colloquially speaking, a popular matching is a matching that would not lose a head-to-head election against any other matching if the agents were allowed to vote between the matchings. 

\subsection{Related Work}

We first review existing work on matchings under preferences in the three-dimensional setting, and then highlight the most important improvements on popular matchings.

\subsubsection{Stability in 3 Dimensions} 
After the introduction of stable matchings by \citet{GS62} and their celebrated algorithm to solve the problem in bipartite graphs, the study of three-dimensional stable matchings was initiated by \citet{Knu76}, who asked about a generalization of stable matchings to triples. Subsequently, \citet{NH91} studied a stable matching variant with three genders, where agents of one gender have a preference list over \textit{pairs} of the other two genders. The goal in this model is to find a set of disjoint triples that is not blocked by any triple outside of it. 
\citet{NH91} and independently \citet{Sub94} were able to prove that it is \NP-complete to decide whether such a three-dimensional stable matching exists. Their result was then generalized by \citet{Hua07}, who incorporated ties and stronger notions of stability, as well as restricted preference structures in this model. 
He showed that all these variants stay \NP-complete as well. \citet{Dan03} identified an even further restricted preference structure that allows for a polynomial-time algorithm for the existence problem.

\subsubsection{3D-Stable Matchings with Cyclic Preferences} One direction proposed by \citet{NH91} was to generalize their work to cyclic preferences. This question lead to a family of papers. Bir\'o and McDermid~\cite{BM10} showed that deciding whether a weakly stable matching exists is NP-complete if preference lists are allowed to be incomplete, and that the same complexity result holds for strong stability even with complete lists. However, the combination of complete lists and weak stability proved to be extremely challenging to solve.

For this setting, Boros et al.~\cite{BGJK04} proved that each instance admits a weakly stable matching for $n \leq 3$, where $n$ is the size of each agent set in the tripartition. Eriksson et al.~\cite{ESS06} later extended this result to $n \leq 4$. Additionally, Pashkovich and Poirrier~\cite{PP20} further proved that not only one, but at least two stable matchings exist for each instance with $n = 5$. By this time, the conjecture on the guaranteed existence of a weakly stable matching in 3D instances with complete cyclic preferences became one of the most riveting open questions in the matching under preferences literature~\cite{Knu76,Man13,Woe13}. Surprisingly, Lam and Plaxton~\cite{LP19} recently disproved this conjecture by showing that weakly stable matchings need not exist for an arbitrary $n$, moreover, it is \NP-complete to determine whether a given instance with complete lists admits a weakly stable matching.

The problem is relevant to applications as well, as shown by the papers of \citet{CJ13}, \citet{RGJTPSH19}, and \citet{MGWJG20}, who all studied 3D-cyclic stable matchings in the context of computer networks, as well as by the work of \citet{BCG20}, who applied it to a Paris apartment assignment problem. Additionally, \citet{EO2018} set up constraint programming models for the problem. They discussed instances where agents of the same class have identical preference lists. This type of preference structure is also called a \textit{master list}. Besides them, \citet{BHKN20} also investigated master lists in the context of 3D stable matchings, and there is a large set of results on 2D stable or popular matchings with master lists in the input~\cite{KNN14, Kam19, IMS08, MR20}.

\subsubsection{Popular Matchings}
The concept of a popular matching corresponds to the notion of a weak Condorcet winner in voting. In the context of matchings it was first introduced by \citet{Gar75} for matching markets with two-sided preferences, and then studied by \citet{AIKM07} in the house allocation problem. Polynomial time algorithms to find a popular matching were given in both settings. These papers inspired a plethora of work on popularity in the house allocation problem. Most importantly, \citet{SM10} extended the model of \citet{AIKM07} with capacities on the houses, while \citet{MI11} studied a weighted variant.%

In the classic two-sided preferences model, it was already noticed by \citet{Gar75} that all stable matchings are popular, which implies that in standard bipartite stable matching instances, popular matchings always exist. In fact stable matchings are the smallest size popular matchings, as shown by \citet{BIM10}, while maximum size popular matchings can be found in polynomial time as well~\cite{HK13,Kav14}.

Only recently \citet{FKPZ19} and  \citet{GMSZ21} resolved the long-standing open question that it is \NP-complete to find a popular matching in a non-bipartite matching instance.%

\subsubsection{Strongly Popular Matchings} A further concept we study is that of a strongly popular matching, corresponding to a strong Condorcet winner, \ie a matching that wins every head-to-head election. This concept was introduced by \citet{BIM10}, who showed that a strongly popular matching in roommates instances exists if and only if it is the unique stable matching. The open question whether a strongly popular matching in a roommates instance with ties can be found in polynomial time was recently answered positively by \citet{BB20}, who observed that a strongly popular matching must be the unique mixed popular matching. Strong popularity was very recently extended to $b$-matchings as well by \citet{KM20}.

\subsubsection{Popularity in 3 Dimensions} \citet{BB20} showed that it is intractable to find a popular partition into sets of at most size three, even if the ranking of all sets by all agents is the same. This however is different from the 3D-cyclic model in both the structure of the preferences, since the agents in their model have a preference list over subsets of size $2$ or $3$, as well as in the structure of the solution, since they allow sets of size $2$ and $3$. Both \citet{BB20} and \citet{LP19} mentioned the 3D-cyclic popular matching problem as an interesting future research direction.

\section{Preliminaries}
\label{sec:prelims}
We now define the notation we use and the problems we investigate in this paper.

\subsection{Input and Output Formats}

\subsubsection{Input and Notation} We are given three sets of \emph{agents} $A,B$, and $C$. We denote by $V = A \cup B \cup C$ the set of all agents and we call $A,B$, and $C$ the \emph{agent classes} of our instance. Further we assume that $\lvert A \rvert = \lvert B \rvert = \lvert C \rvert = n$. 
Each agent in $A$ has a strict preference list over a subset of agents in $B$, each agent in $B$ has a strict preference list over a subset of agents in $C$, and finally, each agent in $C$ has a strict preference list over a subset of agents in $A$. These preference lists define for each agent $x$ a strict order $\succ_x$, which we call the \emph{preference list} of $x$ and say that $x$ finds the agents in $\succ_x$ \emph{acceptable}. For any two agents $y,z$ such that $y \succ_x z$, we say that $x$ \emph{prefers} $y$ to $z$.

\subsubsection{Master Lists} When defining master lists, we use the terminology from the book of Manlove~\cite{Man13}. We say that the preferences of agents in $X \subseteq V$ are \emph{derived from a master list} if there is a master preference list from which the preferences of each $x \in X$ can be obtained by deleting some agents. This means that the preferences might be incomplete, but the relative preferences between acceptable agents are the same in each $\succ_x$, where $x \in X$. We say that an instance is derived from a $k$-master list for $k \in \lbrace 1,2,3\rbrace$ if the preferences of $k$ of the agent classes of our instance are derived from a master list.

\subsubsection{Matchings} A \emph{matching} $M$ is a subset of $A \times B \times C$, such that each agent appears in at most one triple and for each $(a,b,c) \in M$, $a$ finds $b$ acceptable, $b$ finds $c$ acceptable, and $c$ finds $a$ acceptable. If $(a,b,c) \in M$, we also write $M(a) = b, M(b) = c, $ and $M(c) = a.$ If an agent $x$ does not appear in any triple in $M$, we write $M(x) = x$ and say that the agent $x$ is \emph{unmatched}. For convenience in notation we assume that for any agent $x$, $x$ itself appears at the end of $\succ_x$. The preference relation $\succ$ naturally extends to the comparison of two triples by an agent who is in both triples.

\subsection{Optimality Concepts}
\label{sec:optconc}

\subsubsection{Weak and Strong Stability} A triple $t=(a, b, c)$ is said to be a \textit{strongly blocking triple} to matching $M$ if each of $a, b$, and $c$ prefer $t$ to their respective triples in~$M$. Practically, this  means that $a, b$, and $c$ could abandon their triples to form triple $t$ on their own, and each of them would be strictly better off in $t$ than in~$M$. If a matching $M$ does not admit any strongly blocking triple, then $M$ is called a \textit{weakly stable} matching.
Similarly, a triple $t=(a, b, c)$ is called a \textit{weakly blocking triple} if at least two agents in the triple prefer $t$ to their triple in $M$, while the third agent does not prefer her triple in $M$ to~$t$. This means that at least two agents in the triple can improve their situation by switching to $t$, while the third agent does not mind the change. A matching that does not admit any weakly blocking triple is referred as \textit{strongly stable}. By definition, strongly stable matchings are also weakly stable, but not the other way round. Observe that it is impossible to construct a triple $t$ that keeps exactly two agents of a triple equally satisfied, while making the third agent happier, since the earlier two agents need to keep their partners to reach this, which then already defines the triple as one already in~$M$.%

\subsubsection{Weak and Strong Popularity} Given an agent $x$ and two matchings $M$ and $\Mp$, we define 
\[
\vote_x(\Mp, M) = 
    \begin{cases*}1, &if $\Mp(x) \succ_x M(x)$ \\
    0, &if  $\Mp(x) = M(x)$ \\
    -1, &if  $M(x) \succ_x \Mp(x)$
    \end{cases*}
\]
\ie $\vote_x(\Mp, M)$ represents whether the agent $x$ would prefer to be in $\Mp$ or in~$M$. We call $\Mp$ \emph{more popular} than $M$ if 
\[
\Delta(\Mp, M) \coloneqq \sum\limits_{x \in V} \vote_x(\Mp, M) \ge 1
\]\ie if $\Mp$ would win against $M$ in a head-to-head election.
Matching $M$ is called \emph{popular} if no matching is more popular than~$M$.
Using this we can now define the popular matching problem in 3 dimensions.
\dprob{\textsc{3d-cyclic popular matching with incomplete lists (3dpmi)}}{Sets $A,B,C$ with a cyclic preference structure.}{Does a popular matching exist?} 
Further we also study the corresponding verification problem.
\dprob{\textsc{3d-cyclic popular matching verification with incomplete lists (3dpmvi)}}{Sets $A,B,C$ with a cyclic preference structure and a matching~$M$.}{Is $M$ popular?}

The notion of popularity can be strengthened even further to what is commonly referred to as a strongly popular matching. A matching $M$ is \textit{strongly popular} if it is more popular than all other matchings~$\Mp$. It is easy to see that each instance can admit at most one strongly popular matching. Now we can define the problems of existence and verification for a strongly popular matching.

\dprob{\textsc{3d-cyclic strongly popular matching with incomplete lists (3dspmi)}}{Sets $A,B,C$ with a cyclic preference structure.}{Does a strongly popular matching exist?}
\dprob{\textsc{3d-cyclic strongly popular matching verification with incomplete lists (3dspmvi)}}{Sets $A,B,C$ with a cyclic preference structure and a matching~$M$.}{Is $M$ strongly popular?}
If we want to indicate that the preference lists are complete, \ie every agent in $A$ ranks all agents in $B$, every agent in $B$ ranks all agents in $C$, and every agent in $C$ ranks all agents in $A$, we omit the \textsc{i} from the end of the problem name.

\subsubsection{$A \cup B$-Popularity}
Our last optimality concept relies on a recent real application, described by \citet{BCG20} who analyzed the Paris public housing market. In their work, $A$ consists of various housing institutions such as the Ministry of Housing, $B$ is the set of households looking for an apartment, and finally, $C$ contains the social housing apartments that are to be assigned to these households. Institutions have preferences over household-apartment pairs, and households rank the available apartments in their order of preference. One characteristic feature of this application is that apartments are treated as objects without preferences, because they should be matched through the institutions.

Here we will study a restricted variant, listed as one of the three most typical interpretations of the institutions' preferences by \citet{BCG20}: institutions have preferences directly over the households, no matter which apartment they are matched to. This problem setting translates into a 3-dimensional matching instance, where agents in $C$ only have the constraint to be matched to an acceptable agent from $A$, while classes $A$ and $B$ submit preferences over acceptable agents in classes $B$ and $C$, respectively. While \citet{BCG20} focused on the existence of a Pareto optimal solution, here we define popularity for such instances.

Matching $\Mp$ is $A \cup B$-more popular than matching $M$ if 
\[
\sum\limits_{x \in A \cup B} \vote_x(\Mp, M) \ge 1,
\]\ie if $\Mp$ would win against $M$ in a head-to-head election where only agents in $A \cup B$ are allowed to vote. Analogously, we call a matching $M$ $A \cup B$-popular if there is no matching that is $A \cup B$-more popular than~$M$. This definition tallies the votes of each household and institution, but treats apartments as objects. To overcome the technical difficulty of one institution handling more than one apartment and to give a vote to the institution in the decision over each apartment, we can simply clone the institutions as many times as many apartments they oversee.%
 
\subsection{Our Contribution}

\begin{table}[tb]

	\begin{center}
			\resizebox{\columnwidth}{!}{ 
		\begin{tabular}{l| c c | c c|} \cline{2-5}	
			& \multicolumn{2}{c|}{Existence} & \multicolumn{2}{c|}{Verification}   \\ \cline{2-5}
			&incomplete & complete & incomplete & complete \\ \hline
			\multicolumn{1}{|l|}{Popularity} & \NP-h. Theorem~\ref{thrm_np_pop_inc} & ? & \NP-c. Theorem~\ref{thrm_np_pop_v_inc} & ? \\
			\multicolumn{1}{|l|}{Strong Popularity} &\NP-h. Theorem~\ref{thrm_strongpop_np}& ?&\NP-c. Theorem~\ref{th:3dspmvi}& ? \\
			\multicolumn{1}{|l|}{$A \cup B$-Popularity} &\NP-h. Theorem~\ref{thm:aubpop}& $\phantom{-}\in$\P \;Theorem~\ref{thm:aubpol}& ? &$\phantom{-}\in$\P \;Theorem~\ref{thm:aubpol} \\
			\hline
		\end{tabular}
		}	
		\caption{Overview of the complexity results shown in Section~\ref{sec:hardness}. The columns refer to the cases with incomplete and complete lists, respectively. Question marks denote open problems---these are briefly discussed in Section~\ref{sec:openp}.}
		\label{ta:results}
	\end{center}
\end{table}
We provide structural results and a complexity analysis of the aforementioned popular matching problems. First we show in Section~\ref{sec:struct} that several implications from the 2-sided matching world do not hold. In 3 dimensions, stable matchings are not necessarily popular and strongly popular matchings are not necessarily stable. 

Then in Section~\ref{sec:hardness} we turn to the complexity of verifying and computing a popular, strongly popular, or $A \cup B$-popular matching when lists are complete, and show that the defined verification and search problems for all variants except $A \cup B$-popularity verification are \NP-hard. We complement these results with positive ones for $A \cup B$-popularity with complete lists. Table~\ref{ta:results} summarizes our results. 

Following this we investigate instances derived from master lists in Section~\ref{sec:master}, and show that in general for $3$-master lists and $2$-master lists popular matchings do not exist. Finally, in Section~\ref{sec:openp} we list some interesting problems that are left open by this work. Our hardness proofs have been delegated to the Appendix.

\section{Structural Results}
\label{sec:struct}

As a first step we investigate the relations between stability and popularity. In the traditional stable marriage and roommates problems, stable matchings form a subset of popular matchings~\cite{Gar75}. Moreover, if a strongly popular matching exists, then it must be the unique popular matching and also the unique stable matching in the instance~\cite{BIM10}.

First we show that in 3 dimensions, neither kind of stability implies popularity by presenting an instance with a strongly stable matching that is not popular.
\begin{lemma}
	There is an instance $\mathcal I$ of \textsc{3dpmi} with a matching $M$ such that $M$ is strongly stable but not popular. 
	\label{le:ssnotpop}
\end{lemma}
\begin{proof}
	Consider the preference profiles depicted in Figure~\ref{fi:stablenotpopular}. First we prove that the matching $M = \lbrace (a_1,b_1,c_1),(a_2,b_2,c_2),(a_3, b_3,c_3)\rbrace $ is strongly stable. As we observed in Section~\ref{sec:optconc}, at least two agents in a weakly blocking triple must improve their match. There are only 6 possible improvements to $M$: $b_1$ switches to $c_2$, $c_1$ switches to $a_2$, $a_2$ switches to $b_3$, $b_2$ switches to $c_3$, $a_3$ switches to $b_1$, and finally, $c_3$ switches to $a_1$. It is easy to check that no two of these will keep the third agent involved at least as satisfied as she is in~$M$.
	
	However, the matching $M' = \lbrace (a_1, b_2, c_3), (a_2, b_3, c_1), (a_3, b_1, c_2) \rbrace$ is more popular, since all agents except for $\lbrace a_1, b_3, c_2\rbrace $ prefer it to~$M$.
	\begin{figure}[tb]
	\begin{align*}
		&a_1\colon \underline{b_1}, \boldsymbol{b_2}, b_3 &&a_2\colon \boldsymbol{b_3}, \underline{b_2}, b_1 &&a_3\colon \boldsymbol{b_1}, \underline{b_3}, b_2 \\
		&b_1\colon \boldsymbol{c_2}, \underline{c_1}, c_3 &&b_2\colon \boldsymbol{c_3}, \underline{c_2}, c_1 &&b_3\colon \underline{c_3}, c_2, \boldsymbol{c_1} \\
		&c_1\colon \boldsymbol{a_2}, \underline{a_1}, a_3 &&c_2\colon \underline{a_2}, a_1, \boldsymbol{a_3} &&c_3\colon \boldsymbol{a_1}, \underline{a_3}, a_2
	\end{align*}
	\caption{Compact representation of the preferences in Lemma~\ref{le:ssnotpop}. Agent $a_1$ has the preference list $b_1 \succ_{a_1} b_2 \succ_{a_1} b_3$. The triples in $M$ are underlined. Bold font denotes the more popular matching~$M'$.}
	\label{fi:stablenotpopular}
    \end{figure}  
\qed \end{proof}

Secondly we show that for 3-dimensional instances, even strong popularity does not imply weak stability.%
\begin{lemma}
	There is an instance $\mathcal I$ of \textsc{3dspmi} with a matching $M$ such that $M$ is strongly popular but not weakly stable.
	\label{obs:strongpop_weakstab}
\end{lemma}
\begin{proof}
    Consider the preference profiles depicted in Figure~\ref{fig:pref_2} and the matching $M \coloneqq \lbrace (a_1,b_1,c_1),(a_2,b_2,c_2),(a_3, b_3,c_3)\rbrace$.
	As can be easily seen, $M$ is not weakly stable, since $(a_1, b_2, c_3)$ is a strongly blocking triple.

	Matching $M$ is however strongly popular. Assume indirectly that there is a matching $\Mp$ such that $M$ is not more popular than~$\Mp$. The only three agents who can possibly improve are $a_1, b_2,$ and $c_3$, because the remaining 6 agents are matched to their first choice. If all three of them improve in $\Mp$, then $(a_1, b_2, c_3) \in \Mp$, and all possible matchings for the remaining 6 agents match at least 4 of them to an agent who is not their top choice. The other possibility is that at least one of $a_1, b_2,$ and $c_3$ remains in the same triple in $\Mp$ as she was in~$M$. Due to symmetry, we can assume without loss of generality that this agent is $a_1$, and thus, $(a_1, b_1, c_1) \in \Mp$. The only agent who can improve from this point on is~$b_2$, and she must switch to~$c_3$. This $\Mp$ can be only finished by taking $(a_2, b_2, c_3), (a_3, b_3, c_2) \in \Mp$ or by taking $(a_3, b_2, c_3), (a_2, b_3, c_2) \in \Mp$, both of which make only $b_2$ better off, and exactly 3 out of these 6 agents worse off than they were in~$M$. Thus $M$ is strongly popular.
		\begin{figure}[tb]
	    \begin{align*}
		&a_1\colon \boldsymbol{b_2}, \underline{b_1}, b_3 &&a_2\colon \underline{b_2}, b_3, b_1 &&a_3\colon \underline{b_3}, b_2, b_1 \\
		&b_1\colon \underline{c_1}, c_2, c_3 &&b_2\colon \boldsymbol{c_3}, \underline{c_2}, c_1 &&b_3\colon \underline{c_3}, c_2, c_1 \\
		&c_1\colon \underline{a_1}, a_2, a_3 &&c_2\colon \underline{a_2}, a_1, a_3 &&c_3\colon \boldsymbol{a_1}, \underline{a_3}, a_2
	\end{align*}
	\caption{Representation of the preferences in Lemma~\ref{obs:strongpop_weakstab}. The triples in $M$ are underlined and the strongly blocking triple is in bold.}
	\label{fig:pref_2}
	\end{figure}
\qed \end{proof}

Our third result shows that in an instance with complete lists, a strongly popular matching can only be blocked by 
strongly blocking triples.
\begin{lemma}
In an instance $\mathcal I$ of \textsc{3dspm}, a strongly popular and weakly stable matching $M$ is also strongly stable.
\label{lem:strong_pop_strong_stab}
\end{lemma}
\begin{proof}
	Consider a triple $t=(a,b,c)$ and assume that $t$ strongly blocks~$M$. Since $M$ is weakly stable, one of the three agents needs to have the same partner in $t$ and in~$M$. Without loss of generality we assume that this agent is $b$, and thus $b$ and $c$ were matched in $M$ as well. Let $(a, \beta, \gamma), (\alpha,b,c) \in M$ be triples in~$M$. Since $(a,b,c)$ is a strongly blocking triple to $M$, we know that $a \succ_c \alpha$.
	The matching $M' = M \setminus (a, \beta, \gamma) \setminus (\alpha,b,c) \cup (\alpha, \beta, \gamma) \cup (a,b,c)$
	leads to at least two agents, $a$ and $c$, preferring $M'$ to $M$, and at most two agents, $\alpha$ and $\gamma$, preferring $M$ to~$M'$. This contradicts the assumption that $M$ was strongly popular.
\qed \end{proof}

\section{\NP-Hardness Results}
\label{sec:hardness}

In this section we prove hardness for all our problems with incomplete lists, except for $A \cup B$-popularity verification. We also show that $A \cup B$-popularity can be verified and an $A\cup B$-popular matching can be found in polynomial time if preference lists are complete. For a structured summary of these results, please consult Table~\ref{ta:results}.

\subsection{Popularity}
We start by showing that it is \NP-hard to determine whether an instance with incomplete lists admits a popular matching. For this we use a restricted, but still \NP-complete variant of \textsc{3sat}, known as \textsc{(2,2)-e3-sat}~\cite{BKS03}.
\dprob{\textsc{(2,2)-e3-sat}}{A set $X$ of variables and a set $\mathcal{C}$ of clauses of size exactly $3$ such that each variable appears in exactly two clauses in positive form and in exactly two clauses in negative form.}{Is there a satisfying assignment for $\mathcal{C}$?}

\begin{restatable}{thm}{popnphard}
It is \NP-hard to decide whether a \textsc{3dpmi} instance admits a popular matching, even if each agent finds four other agents acceptable. This holds even if the preferences are derived from a $2$-master list.
	\label{thrm_np_pop_inc}
\end{restatable}

In order to state \NP-completeness instead of \NP-hardness, we would have to prove that \textsc{3dpmvi} is polynomial-time solvable. However we can show that this is problem is computationally intractable as well.
\begin{restatable}{thm}{popvnphard}
It is \NP-complete to decide whether a given \textsc{3dpmvi} instance with a matching $M$ admits a matching that is more popular than $M$. This holds even if the preferences are derived from a $1$-master list.
	\label{thrm_np_pop_v_inc}
\end{restatable}

\subsection{Strong Popularity}

Next we show that it is also \NP-hard to find a strongly popular matching and to verify whether a given matching is strongly popular. For this we reduce from the problem of finding a perfect matching in a 3D-cyclic matching instance without preferences, shown to be \NP-complete by \citet{GJ79}. 
\dprob{\textsc{perfect 3d-cyclic matching}}{Sets $A,B,C$ with cyclic acceptability relations.}{Does a perfect matching exist?}
\begin{restatable}{thm}{strongpopnphard}
It is \NP-hard to determine whether a given \textsc{3dspmi} instance admits a strongly popular matching. This holds even if the preferences are derived from a $2$-master list.
	\label{thrm_strongpop_np}
\end{restatable}

A slightly modified version of the proof implies that \textsc{3dspmvi} is also computationally intractable. 
\begin{restatable}{thm}{strongpopvnphard}
    It is \NP-complete to decide whether a given \textsc{3dspmvi} instance admits a matching $M'$ such that $M$ is not more popular than~$M'$. This holds even if the preferences are derived from a $2$-master list.
    \label{th:3dspmvi}
\end{restatable}

\subsection{$A \cup B$-Popularity}

Finally we turn to the application-motivated variant of our problem and show that computing a matching that is $A \cup B$-popular, \ie it does not lose any head-to-head election where only agents in $A \cup B$ can vote, is \NP-hard as well. We reduce from the problem of finding a popular matching in a bipartite graph with one side having strict preferences and the other side either having a tie or strict preferences, which was shown to be \NP-complete by \citet{CHK17}. 
\dprob{\textsc{popular matching with one-sided ties}}{A bipartite graph $G = (U \cup W, E)$, for each $u \in U$ a strict preference list over its neighbors in $W$, for each $w \in W$ either a strict preference list or a preference list containing a single tie over its neighbors in~$U$.}{Does $G$ admit a popular matching?}
\begin{restatable}{thm}{abpop}
It is \NP-hard to decide whether a \textsc{3dpmi} instance admits an $A \cup B$-popular matching.
\label{thm:aubpop}
\end{restatable}
Interestingly enough the problem becomes easy with complete lists.
\begin{restatable}{thm}{abcomp}
Both verifying $A \cup B$-popularity and computing an $A \cup B$-popular matching in a \textsc{3dpm} instance $\mathcal{I}$ can be done in linear time.
\label{thm:aubpol}
\end{restatable}%
\begin{proof}
From $\mathcal{I}$ we construct two house allocation instances with one-sided preferences, $\mathcal{I}_A := (A, B, (\succ_a)_{a \in A})$ and $\mathcal{I}_B := (B, C, (\succ_b)_{b \in B})$. We will show that $\mathcal{I}$ admits a popular matching if and only if both $\mathcal{I}_A$ and $\mathcal{I}_B$ admit a popular matching.

First assume that $\mathcal{I}$ admits a popular matching $M$. Using this we now construct the two matchings, $M_A \coloneqq \{(a, M(a) \mid a\in A\}$ in $\mathcal{I}_A$ and $M_B \coloneqq \{(b, M(b) \mid b\in B\}$ in $\mathcal{I}_B$. Without loss of generality assume that $M_A$ is not popular in $\mathcal{I}_A$ and let $M_A'$ be the more popular matching. It is easy to see that the matching $\{a, M_A'(a), M(M_A'(a)) \mid a\in A\}$ is also more popular than $M$ in $\mathcal{I}$. 

If $\mathcal{I}_A$ and $\mathcal{I}_B$ admit popular matchings $M_A$ and $M_B$, respectively, then the matching $\{a, M_A(a) M_B(M_A(a)) \mid a \in A\}$ is clearly popular in~$\mathcal{I}$.

This immediately yields a linear time algorithm for finding an $A \cup B$-popular matching and verifying whether a matching is $A \cup B$-popular matching, since popular matchings in house allocation instances can be found and verified in linear time as shown by \citet{AIKM07}.
\end{proof}

\section{Master Lists}
\label{sec:master}

Now we turn to studying the problem of computing a popular matching in instances with preferences derived from master lists. Examples of real-life applications of master lists occur in  resident matching programs~\cite{BIS11}, dormitory room assignments~\cite{PPR08}, cooperative download applications such as BitTorrent~\cite{ALMO08}, and 3-sided networking services~\cite{CJ13}. 
Even though the presence of a master list usually simplifies stable matching problems and warrants that a solution exists and it is easy to find~\cite{IMS08,OMa07}%
, here we show that for 3-dimensional popular matchings, instances with master lists tend to admit no popular matching at all. This observation is aligned with what has already been shown by the Condorcet paradox, possibly the first example for the non-existence of a weak majority winner.%

\subsection{3-Master List}
First we show that an instance derived from a $3$-master list has no popular matching if there are at least 3 agents per class.

\begin{restatable}{thm}{mlnopop}
A \textsc{3dpm} instance derived from a 3-master list has no popular matching if $n \geq 3$. 
	\label{thrm:ml_nopop}
\end{restatable}
\begin{proof}
	Let $M$ be a maximal matching (otherwise the matching is trivially not popular) and let $ (a_i,b_i,c_i),(a_j,b_j,c_j), (a_k,b_k,c_k)\in M$ be three disjoint triples. 
	Without loss of generality we can assume that $a_i \succ_c a_j \succ_c a_k$.
	We will now distinguish two cases. First assume that one of 
	\begin{itemize}
	    \item $b_j \prec_a b_k \prec_a b_i$;
	    \item $b_k \prec_a b_i \prec_a b_j$;
	    \item $b_i \prec_a b_j \prec_a b_k$
	\end{itemize}	
	holds, \ie the ranking of the three agents in $B$ is `reversed' compared to the ranking of the agents in $A$ they are matched to. Then the matching $\Mp$ resulting from replacing the triples  $\lbrace (a_i,b_i,c_i),(a_j,b_j,c_j), (a_k,b_k,c_k) \rbrace$ by the triples $\lbrace (a_k,b_i,c_i),(a_i,b_j,c_j), (a_j,b_k,c_k)\rbrace$ in $M$ is more popular, since two of $a_i, a_j, a_k$  (as can be seen by the two $\prec_a$) and $c_j, c_k$ prefer $\Mp$, while only two agents are against $\Mp$. 	
	
    For the second case, we can assume that one of 
	\begin{itemize}
	    \item $b_j \succ_a b_k \succ_a b_i$;
	    \item $b_k \succ_a b_i \succ_a b_j$;
	    \item $b_i \succ_a b_j \succ_a b_k$
	\end{itemize} holds, \ie the ranking of the agents in $B$ is cyclically shifted from the ranking of the agents in $A$. Now we construct matching $\Mp$ by replacing the triples $\lbrace (a_i,b_i,c_i),(a_j,b_j,c_j), (a_k,b_k,c_k) \rbrace$ by the triples $\lbrace (a_i,b_k,c_j),(a_j,b_i,c_k), (a_k,b_j,c_i)\rbrace$ in $M$.
	Since $b$ is sorted, two agents in $A$, $c_j$ and $c_k$, and the agent in $b$ who was previously matched to the worst of the three agents in $C$ prefer their partner in $\Mp$ to their partner in~$M$. Thus $\Mp$ is more popular than $M$, and therefore no popular matching exists in this instance. 
\qed \end{proof}

For the sake of completeness, we remark that for $n \leq 2$, all perfect matchings in a \textsc{3dpm} instance derived from a 3-master list are trivially popular. 

Interestingly, \citet{EO2018} were able to show that instances derived from a $3$-master list have exponentially many stable matchings, so Theorem~\ref{thrm:ml_nopop} 
shows a stark contrast between stability and popularity in three-dimensional cyclic matching.

\subsection{2-Master List}
In the spirit of Theorem~\ref{thrm:ml_nopop}, we can also show that even if only the preferences in $A$ and $B$ are derived from a master list, no popular matching exists if there are more than four agents in each of the three classes. 
\begin{restatable}{thm}{mlnoex}
In an instance of \textsc{3dpm} derived from a $2$-master list, no popular matching exists if $n \ge 5$.
	\label{thm:ml_nonex}
\end{restatable}

\begin{proof}
	Without loss of generality we can assume that classes $A$ and $B$ each have a master list, and consider a matching $M$. Let the rankings for $B$ and $C$ be $b_1 \succ_a \dots \succ_a b_n$ and $c_1 \succ_b \dots \succ_b c_n$, respectively. For any $\gamma \in \lbrace a,b,c\rbrace$ let $\gamma_i \oplus 1 = \gamma_{(i - 1 \mod n)}$. Intuitively, the $\oplus$ operation takes one step up on the list of agents in a class, and if it is applied to the first agent in the class, then it jumps to the last agent. Now we compare the matching that consists of triples in the form
	$(a_i, M(a_i)\oplus 1,M(M(a_i)\oplus 1) \oplus 1 )$, \ie we cyclically shift up the agents in $B$ and $C$. 
	In this operation, all agents in $A$ except for the agent matched to $b_1$, and all agents in $B$ except for the agent matched to $c_1$ improve. Thus at least $2n-2$ agents improve and at most $n+2$ agents receive a worse partner than in~$M$. Therefore if $n \geq 5$, then $M$ was not popular.
\qed \end{proof}

For the sake of completeness we elaborate on the case of instances with $n \leq 4$. Firstly, for $n \leq 2$, it is easy to see that there is at least one popular matching in each \textsc{3dpm} instance derived from a 2-master list. Instances with $n=3$ and $n=4$ can be yes- and no-instances as well. Since the input size is constant, even iterating through all matchings and checking each of them for popularity delivers a polynomial-time algorithm to decide whether a given instance admits a popular matching.

\subsection{1-Master List}
A result analogous to Theorems~\ref{thrm:ml_nopop} and~\ref{thm:ml_nonex} is unlikely to exist if only one agent class is equipped with a master list. Instead, we give a characterization for strongly popular matchings in instances derived from a $1$-master list with complete lists. This characterization also immediately gives us a linear time algorithm to find and verify a strongly popular matching in these instances. The analogous questions for popularity are discussed as open problems in Section~\ref{sec:openp}.

\begin{restatable}{thm}{strongpop1ml}
In an \textsc{3dspm} instance derived from a $1$-master list, a matching is strongly popular if and only if all agents without a master list are matched to their top choice.
	\label{thrm:strong_pop_1ml}
\end{restatable}
\begin{proof}
	Without loss of generality we can assume that class $A$ is equipped with a master list. Let $M$ be a matching that assigns all agents in $B$ and $C$ their top choice. Let $M^\prime$ be any other matching and let $(a_i, b_j, c_k) \in \Mp$. In order to prove that $M$ is more popular than $M'$, we first distinguish three cases.
	\begin{itemize}
	    \item If $M(a_i) = b_j$, then $\vote_{a_i}(\Mp, M) + \vote_{b_j}(\Mp, M) + \vote_{c_k}(\Mp, M) \le 0$ follows from $b_j$ and $c_k$ being matched to their top choice in~$M$.
	    \item If $b_j \succ_a M(a_i) $, then one of $b_j$ and $c_k$ has to be matched to a different partner in $\Mp$ and we get $\vote_{a_i}(\Mp, M) + \vote_{b_j}(\Mp, M) + \vote_{c_k}(\Mp, M) \le 1 - 1 = 0$. 
	    \item If $M(a_i) \succ_a b_j$, then $\vote_{a_i}(\Mp, M) + \vote_{b_j}(\Mp, M) + \vote_{c_k}(\Mp, M) \le -1 - 1 = -2$.
	\end{itemize}
	We also know that any matching changing only the partners of agents in $B$ and $C$, but not in $A$, would trivially be less popular than~$M$. Thus, for $M'$, the last case needs to occur at least once, since all the agents in $A$ have the same preference list.
	All in all we get that $\Delta(\Mp, M) \le -2$ for all $M'$, and therefore $M$ is strongly popular. 
		
	For the opposite direction, we assume that the matching $M$ does not assign all agents in $B$ and $C$ their top choice. Without loss of generality assume that there is an agent $b_i \in B$ such that $b_i$'s top choice is $c_i$, yet $M(b_j) = c_i$. So we assume that $(a_i, b_j, c_i), (a_j, b_i, c_j) \in M$. Next consider the matching that is created by swapping $b_j$ and $b_i$, \ie by removing $(a_i, b_j, c_i), (a_j, b_i, c_j)$ from $M$ and adding $(a_i, b_i, c_i), (a_j, b_j, c_j)$ to $M$. If $b_i \succ_a b_j$, then $a_i$ and $b_i$ improve, while $a_j$ and $b_j$ get worse. Similarly, if $b_i \prec_a b_j$, then $a_j$ and $b_i$ improve, while $a_i$ and $b_j$ get worse. Thus $M$ was not strongly popular. 
\qed \end{proof}

\section{Open Problems}
\label{sec:openp}
Our work leaves three important questions open. The first, related to our results in Section~\ref{sec:hardness}, is the complexity of our problems with regard to complete preference lists. The technique of introducing so-called `boundary dummy-agents' of \citet{LP19} for showing hardness with complete lists for the stable matching problem does not seems to be applicable for popularity, since the presence of blocking edges if an agent is matched below her `boundary' does not restrict the set of popular matchings. Thus in order to reduce either from the problem with incomplete lists or from a separate problem altogether, a new technique might be needed.

Related to this is also the complexity of verifying whether a matching is $A \cup B$-popular with incomplete lists. Due to the inherent hardness of computing weight-optimal or even perfect matchings in 3 dimensions, we conjecture that this problem is \NP-complete as well. 

The third open problem, in case the problem of finding a popular matching with complete lists turns out to be intractable, is that of finding a popular matching in a \textsc{3dpm} derived from a 1-master list. Here, as opposed to instances derived from 2- and 3-master lists, popular matchings can exist. Interestingly enough the structure of these popular matchings seems to be quite limited, since in any situation the agents with the master list could be 'shifted up' to generate a matching that is more popular for at least $n-1$ agents, similarly as we argued in the proof of Theorem~\ref{thm:ml_nonex}. This might lead to results similar to the classification of popular matchings in house allocation instances by \citet{AIKM07}---for instance, if there is a perfect matching $M$ and agents $b_i, b_k \in B$ and $c_j, c_l \in C$ with $l > k \ge j > i$ such that $c_j \succ_{b_i} M(b_i)$ and $c_l \succ_{b_j} M(b_l)$, then $M$ is not popular. Other results of this type might pave the path to a full classification of popular matchings in these instances.

\bibliography{3dpopnocomm}
\bibliographystyle{splncs04nat}
\newpage

\section{Appendix}
First we present an instance of \textsc{3dspm} and show that it admits no popular matching. This instance will come handy later, in the proof of Theorem~\ref{thrm_np_pop_inc}.

\begin{restatable}{obs}{removeagent}
The instance $\mathcal{I}$ of \textsc{3dspm} obtained by removing any agent from an instance with $3$ agents per class and a 3-master list admits a popular matching.
    \label{obs:remove_agent}
\end{restatable}

\begin{proof}
    Without loss of generality assume that at least one agent has been removed from $A$ and that $A$ is the class of smallest cardinality. We distinguish two cases. 
    \begin{itemize}
        \item If $|A|=1$, then we can put all three top choices as the only triple in our matching. Improving the situation of any agent would lead to the respective top choice agent of the same class to be unmatched which would cancel out their votes. The matching is thus popular. 
        \item If $|A|=2$, we can include all top choices and all second choices as our two triples. By enumerating all possible matchings it can be seen that this matching is popular.\qed
    \end{itemize}
 \end{proof}
\popnphard*
\begin{proof}
	We reduce from \textsc{(2,2)-e3-sat}.
	\subsubsection{Construction} Let $X = \lbrace x_1, \dots, x_n\rbrace$ be the set of variables and $\mathcal{C}$ be the set of clauses. We first define the sets of agents $A,B,C$ of our \textsc{3dmpi} 
	instance.
	For each clause $\varphi = \lbrace x_i, x_j, x_k \rbrace$, where $x_i$ might be in  either positive or negative form, we add nine agents
	\begin{itemize}
	    \item three \emph{variable agents} $a_i^\varphi, a_j^\varphi, a_k^\varphi$ in $A$;
	    \item three \emph{dummy agents} $b_1^\varphi, b_2^\varphi, b_3^\varphi$ in $B$;
	    \item three \emph{dummy agents} $c_1^\varphi, c_2^\varphi, c_3^\varphi$ in $C$.
	    
	\end{itemize}
	
	For each variable $x_i \in X$ we include twelve agents 
	\begin{itemize}
	\item two agents $a_1^i, a_2^i$ in $A$;
	\item six agents $b_1^i, b_2^i, b_1^{i,+}, b_1^{i,-}, b_2^{i,+}, b_2^{i,-}$ in $B$;
	\item four agents $c_1^{i,+}, c_1^{i,-}, c_2^{i,+}, c_2^{i,-}$ in $C$.
	\end{itemize}
	
	Next we define our preference lists. 
	\begin{itemize}
	    \item For any clause $\varphi$ and any clause agent $a_i^\varphi$ such that $x_i$ appears is positive form in $\varphi$, we define the preference list to be $b_1^{i,+}\succ b_2^{i,+} \succ b_1^\varphi \succ b_2^\varphi \succ b_3^\varphi$.
	    \item For any clause $\varphi$ and any clause agent $a_i^\varphi$ such that $x_i$ appears is negative form in $\varphi$, we define the preference list to be $b_1^{i,-} \succ b_2^{i,-}\succ b_1^\varphi \succ b_2^\varphi \succ b_3^\varphi$.
	    \item For any $b_m^\varphi$ with $m \in \lbrace 1,2,3\rbrace$ we add the list
	$c_1^\varphi \succ c_2^\varphi \succ c_3^\varphi$.
	\item Lastly for any $c_m^\varphi$, with $x_i, x_j$, and $x_k$ being the variables that appear either in positive or negative form in $\varphi$ such that $i < j < k$, we add the list $a_i^\varphi \succ a_j^\varphi \succ a_k^\varphi$.
	\end{itemize}
	Note that this implies that all the clause and dummy agents belonging to one clause form a sub-instance derived from a $3$-master list. 
	Thus following Theorem~\ref{thrm:ml_nopop} and Observation~\ref{obs:remove_agent}, 
	in any popular matching at least one of the clause agents needs to be matched to a non-dummy agent.

	Next for the variable gadget of any variable $x_i$ we define the following preference lists.
	\begin{itemize}
	    \item Agent $a_1^i$ receives the preference list  $b_2^i \succ b_1^i$.
	    \item Agent $a_2^i$ receives the preference list $b_1^i \succ  b_2^i $.
		\item Agent $b_1^i$ receives the preference list $c_2^{i,-} \succ c_2^{i,+} $.
		\item Agent $b_2^i$ receives the preference list $c_1^{i,+} \succ c_1^{i,-} $.
		\item Agent $b_1^{i,+}$ receives the preference list $ c_1^{i,+}  $.
		\item Agent $b_1^{i,-}$ receives the preference list $c_1^{i,-}  $.
		\item Agent $b_2^{i,+}$ receives the preference list $ c_2^{i,+}$.
		\item Agent $b_2^{i,-}$ receives the preference list $c_2^{i,-}$.
	\end{itemize}
	
	Furthermore we call $\varphi^+, \psi^+$ the clauses where $x_i$ appears in positive form, and $\varphi^-, \psi^-$ the clauses where $x_i$ appears in negative form and turn to the preferences of the variable agents in $C$.
	\begin{itemize}
		\item For the agent $c_1^{i,+}$ we add the preference list  $a_i^{\varphi^+} \succ a_i^{\psi^+} \succ a_1^i$.
		\item For the agent $ c_1^{i,-}$ we add the preference list $a_i^{\varphi^-} \succ a_i^{\psi^-} \succ a_2^i$.
		\item For the agent $c_2^{i,+}$ we add the preference list $a_i^{\varphi^+} \succ a_i^{\psi^+} \succ a_2^i$.
		\item For the agent $c_2^{i,-}$ we add the preference list $a_i^{\varphi^-} \succ a_i^{\psi^-} \succ a_1^i$.
	\end{itemize}
	Note that in this construction the relative order of the preferences in $B$ and $C$ is the same, thus the preferences are subsets of an instance derived from a $2$-master list.
	For a representation of the construction in the variable gadget, see Figure~\ref{fig:reduction_gadget}.
	\begin{figure}[t]
	    \centering
	    \includegraphics{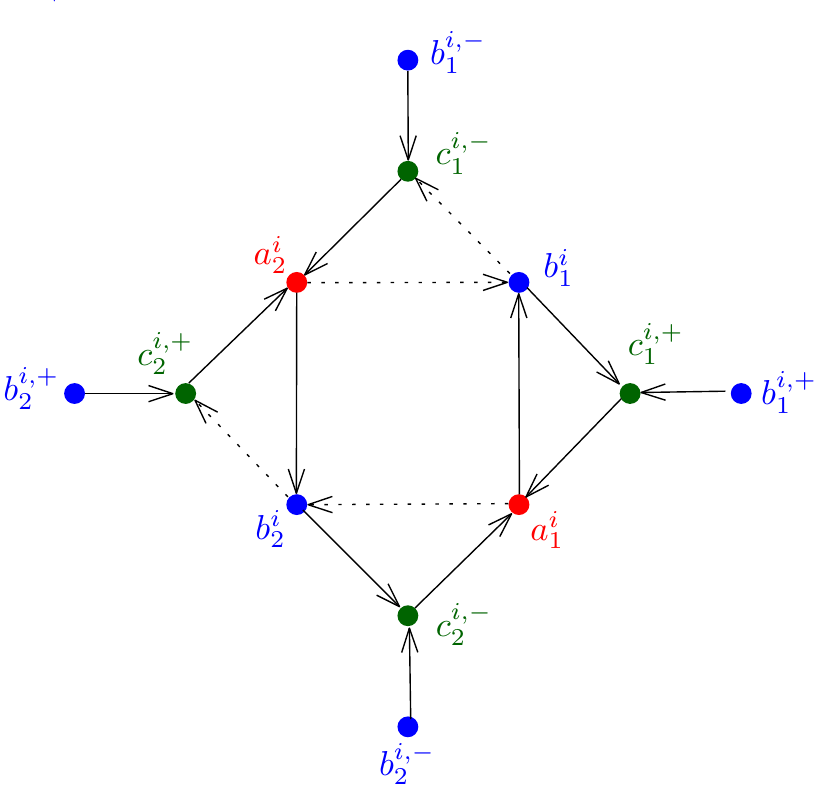}
	    \caption{The figure represents the clause gadget in the proof of Theorem~\ref{thrm_np_pop_inc}. The vertices denote the agents, with vertices of the same color belonging to the same class. Solid edges leaving a vertex represent the first choice of the corresponding agent, whereas dotted edges represent the second choice.}
	    \label{fig:reduction_gadget}
	\end{figure}	%
	
	\subsubsection{$\Rightarrow$} We first assume that the \textsc{(2,2)-e3-sat} instance is satisfiable and $\Phi$ is a satisfying assignment. We now show how to construct a popular matching~$M$.
	\begin{itemize}
	    \item 	For any variable $x_i$ that is set to true in $\Phi$ and appears in positive form  in the clauses $\varphi$ and $\psi$ we include the triples $(a_2, b_2, c_1^{i,-}), (a_1, b_1, c_2^{i,-} ), (a_i^\varphi, b_2^{i,+}, c_2^{i,+})$, and $(a_i^\psi, b_1^{i,+}, c_1^{i,+})$.
	    \item If $x_i$ is set to false in $\Phi$ and appears in negative form in the clauses $\varphi$ and $\psi$ we  include the triples $(a_2, b_1, c_2^{i,+}), (a_1, b_2, c_1^{i,+} ), (a_i^\varphi, b_2^{i,-}, c_2^{i,-})$, and $(a_i^\psi, b_1^{i,-}, c_1^{i,-})$. 
	    \item For any clause $\varphi$ where two variables $x_i, x_j$ are unsatisfied, we add the triples $(a_i^\varphi, b_1, c_1)$ and $(a_j^\varphi, b_2, c_2)$.
	    \item For any clause $\varphi$ with one variable $x_i$ unsatisfied we add the triple $(a_i^\varphi, b_1, c_1)$.
	\end{itemize}
	Now all we are left to do is to show that $M$ is popular. 
	First, we already know from Lemma~\ref{obs:remove_agent} that each matching in each clause gadget is popular if one does not match an additional $a_i^\varphi$ 
	to some agent in the clause gadget. If $a_i^\varphi$ would get matched to some agent in the clause gadget, the matching will get worse for $a_i^\varphi$, $b_{1/2}^{i, \pm}$, and $c_{1/2}^{i, \pm}$ (depending on how the matching is constructed), while at most three agents in the clause gadget can improve. Furthermore any perfect matching $\Mp$ in the variable gadget is not more popular 
	, since switching $a_1^i$ or $b_1^i$ to $\Mp$ increases the rank of one of them, while it decreases the rank of the other one compared to~$M$. If the matching $\Mp$ in the variable gadget is not perfect, at least two of the $a_{1/2}^i$ or $b_{1/2}^i$ are now unmatched and matching $c_{1/2}^{i, \pm}$ to any variable gadget would also unmatch two of the variable dummies, leading to $M$ being preferred by at least one more agent over $\Mp$. 
	Finally matching agents $a_i^\varphi$ and $a_i^\psi$ to the respective other $b_{1/2}^{\pm}$ would make two agents happier and two agents unhappier, thus also not leading to a more popular matching. Therefore in any matching $\Mp$ for any agent who prefers $\Mp$ to $M$ there is at least one (unique) agent who prefers $M$ to $\Mp$, which implies that $M$ is popular.	
	
	\subsubsection{$\Leftarrow$} Next assume that we are given a popular matching $M$. 
	Now we make two observations. 
	\begin{itemize}
	    \item For any clause $\varphi$, at least one agent $a_i^\varphi$ has to be matched to her variable gadget, since otherwise following Theorem~\ref{thrm:ml_nopop} the matching could not be popular, because the agents in the clause gadget are all derived from a master list. 
	    \item Now assume that there is a variable gadget where both at least one of $b_1^{i,-}$ and $b_2^{i,-}$ as well as at least one of $b_1^{i,+}$ and $b_2^{i,+}$ 
	    are matched to clause gadgets. Then without loss of generality $a_1^i$ and $b_1^i$ are unmatched in $M$. Since $M$ is popular, it has to be maximal and thus $c_1^{i,-}$ is matched to some $a_i^\varphi$ and two agents $b_k^\varphi, c_k^\varphi$ are unmatched. Taking the triples $(a_1^i,b_1^i, c_1^{i,-})$ and $(a_i^\varphi,b_k^\varphi, c_k^\varphi )$ improves the matching for $a_1^i,b_1^i, b_k^\varphi, c_k^\varphi$ and makes it worse for $c_1^{i,-}, b_1^{i,-},a_i^\varphi$. Thus $M$ could not have been popular.
	\end{itemize}  
	Therefore for each variable, $M$ matches either only $b_1^{i,-}$ and $b_2^{i,-}$ or  $b_1^{i,+}$ and $b_2^{i,+}$ to their respective clause gadgets and since each clause gadget has to be matched, this implies that we can construct a satisfying assignment.
\qed \end{proof}
\strongpopnphard*
\begin{proof}
We reduce from \textsc{perfect 3d-cyclic matching}. 

\subsubsection{Construction} Assume we are given a \textsc{perfect 3d-cyclic matching} instance $\mathcal{I}$ with sets $A_0 = \lbrace a_1, \dots, a_n \rbrace$, $B_0 = \lbrace b_1, \dots, b_n \rbrace$, and $C_0 = \lbrace c_1, \dots, c_n \rbrace$. For our \textsc{3dspmi} instance, we create a copy of each of the three classes, $A_0' = \lbrace a'_1, \dots, a'_n \rbrace$, $B_0' = \lbrace b'_1, \dots, b'_n \rbrace$, and $C_0' = \lbrace c'_1, \dots, c'_n \rbrace$ and we set $A = A_0 \cup A_0'$, $B = B_0 \cup B_0'$, and $C = C_0 \cup C_0'$.

Next we turn to the preferences. For each agent $x \in A_0 \cup B_0 \cup C_0$ let $n^x_1, \dots, n^x_k$ be her set of acceptable agents in $\mathcal{I}$ in an arbitrary order, such that the relative order between all agents of one class is the same 
Just as in Theorem~\ref{thm:ml_nonex}%
, we take the indices modulo $n$ and as such take $n+1 = 1$. 
\begin{itemize}
    \item For any $a_i \in A_0$ we create the preference list $b_i' \succ_{a_i} n^{a_i}_1 \succ_{a_i} \dots \succ_{a_i} n^{a_i}_k$. 
    \item For any $a_i' \in A_0'$ we create the preference list only consisting of $b_i'$.
    \item For any $b_i \in B_0$ we create the preference list  $n^{b_i}_1 \succ_{a_i} \dots \succ_{b_i} n^{b_i}_k$.
    \item For any $b_i' \in B_0'$ we create the preference list $c_i' \succ_{b_i'} c_{i+1}'$.
    \item For any $c_i \in C_0$ we create the preference list  $n^{c_i}_1 \succ_{c_i} \dots \succ_{c_i} n^{c_i}_k$.
    \item For any $c_i' \in C_0'$ we create the preference list $a_i \succ_{c_i'} a_{i-1}'$. 
\end{itemize}

\subsubsection{$\Rightarrow$}
First assume that $\mathcal{I}$ admits no perfect matching. We then show that the matching $M = \lbrace (a_i, b_i', c_i') \mid a_i \in A_0\rbrace$ is strongly popular. 
Let $\Mp$ be a matching different from~$M$.
For any $i \in [n]$, we define the vertex set $X_i = \lbrace a_i, a_i', \Mp(a_i), \Mp(\Mp(a_i)), b_i',c_i'\rbrace$ and also $\vote_i(\Mp) = \sum_{x\in X_i} \vote_x(\Mp, M)$. 
By the definition of popularity it holds that $\vote(\Mp, M) = \sum_{i = 1}^n \vote_i(\Mp)$.
Now we can distinguish four cases, based on whether $a_i$ and $a_i'$ are unmatched or matched to $b_i'$. 
\begin{itemize}
    \item If $\Mp(a_i) = b_i'$, then it holds that $\vote_i(\Mp) = 0$.
    \item If $\Mp(a_i) = a_i$, then it is easy to see that $\vote_i(\Mp) \le -2$.
    \item Otherwise if $\Mp(a_i) \neq b_i'$ and $\Mp(a_i') = b_i'$, then $a_i', \Mp(a_i)$, and $\Mp(\Mp(a_i))$ prefer $\Mp$ to $M$, while the rest of $X_i$ prefers $M$ to $\Mp$, which results in $\vote_i(\Mp) = 0$.
    \item If $\Mp(a_i) \neq b_i'$ and $\Mp(a_i') = a_i'$, then $  \Mp(a_i)$ and $\Mp(\Mp(a_i))$ prefer $\Mp$ to $M$, while the rest of $X_i$ prefers $M$ to $\Mp$, which results in $\vote_i(\Mp) \le -1.$
\end{itemize}
Since $\Mp \neq M$, $\Mp(a_i) \neq b_i'$ has to hold for at least one $a_i \in A_0$. However since no perfect matching exists, there has to be an agent $a_i \in A$ not matched to an agent in $B_0$. From this follows that either  $\Mp(a_i) = a_i$ or  $M(a_i') = a_i'$  has to hold if this $a_i$ is matched to $b_i'$, 
which would imply that $\vote_i(\Mp) < 0$. Therefore we get that $\vote(\Mp, M) = \sum_{i = 1}^n \vote_i(\Mp) < 0$ and thus $M$ is strongly popular.

\subsubsection{$\Leftarrow$}
For the other direction assume that $\mathcal{I}$ admits a perfect matching $\Mp_0$ and consider the two matchings $M =\lbrace (a_i, b_i', c_i') \mid a_i \in A_0\rbrace$ and $\Mp = \Mp_0 \cup \lbrace a_i', b_i', c_{i+1}' \rbrace$. First it is easy to see that for any $i \in [n]$, all of $a_i ,b_i', c_i'$ prefer $M$ to $\Mp$ while $b_i,c_i,$ and $a_i'$ prefer $\Mp$ to $M$. Thus neither $M$ nor $\Mp$ is strongly popular.

Now assume that $\Mp$ is a matching different from $M$. Let $\vote_i$ be as in the previous case. If $\Mp(a_i) = b_i'$, then nothing changes for the agents in $X_i$ and thus $\vote_i(\Mp) = 0.$ If $\Mp(a_i) = a_i$, then we get $\vote_i(\Mp) \le -2$ and if $\Mp(a_i) \in B_0$, we get $\vote_i(\Mp) \le 0$. Thus the matching $\Mp$ was not strongly popular either and therefore no strongly popular matching exists.
\qed \end{proof}

\strongpopvnphard*
\begin{proof}
This immediately follows from the proof of Theorem~\ref{thrm_strongpop_np} by setting $M = \lbrace (a_i, b_i', c_i') \mid a_i \in A_0\rbrace.$ The membership in \NP\ follows from the fact that any matching that is not more popular than $M$ serves as a witness.
\qed \end{proof}
\popvnphard*
\begin{proof}
Here membership in \NP\ immediately follows from the fact that any matching that is more popular than $M$ serves as a witness, since it can be computed in polynomial time whether a matching is more popular than a given matching.
We reduce from \textsc{perfect 3d-cyclic matching} and show how to modify the proof of Theorem~\ref{thrm_strongpop_np}.

\subsubsection{Construction} We again assume we are given a \textsc{perfect 3d-cyclic matching} instance $\mathcal{I}$ with classes $A_0 = \lbrace a_1, \dots, a_n \rbrace$, $B_0 = \lbrace b_1, \dots, b_n \rbrace$, and $C_0 = \lbrace c_1, \dots, c_n \rbrace$. Without loss of generality we assume that $n$ is odd. For our \textsc{3dspmi} instance, we create a copy of each of the three sets, $A_0' = \lbrace a'_1, \dots, a'_n \rbrace$, $A_0'' = \lbrace a''_1, \dots, a''_n \rbrace$, $B_0' = \lbrace b'_1, \dots, b'_n \rbrace$, $B_0'' = \lbrace b''_1, \dots, b''_n \rbrace$, $C_0' = \lbrace c'_1, \dots, c'_n \rbrace$, and $C_0'' = \lbrace c''_1, \dots, c''_n \rbrace$, and we set $A = A_0 \cup A_0' \cup A_0''$,  $B = B_0 \cup B_0'\cup B_0''$, and $C_0 \cup C_0'\cup C_0''$.

Next we turn to the preferences. For any agent $x \in A_0 \cup B_0 \cup C_0$ let $n^x_1, \dots, n^x_k$ be their set of her acceptable agents in $\mathcal{I}$ in some arbitrary order, such that the relative order between all agents of one class is the same. 
\begin{itemize}
    \item For any $a_i \in A_0$ we create the preference list $b_i' \succ_{a_i} n^{a_i}_1 \succ_{a_i} \dots \succ_{a_i} n^{a_i}_k$. 
    \item For any $a_i' \in A_0'$ we create the preference list only consisting of $b_i'$.
    \item For any $a_i'' \in A_0''$ we create the preference list consisting of $b_i \succ_{a_i''} b_i''$.
    \item For any $b_i \in B_0$ we create the preference list  $c_i'' \succ_{b_i} n^{b_i}_1 \succ_{b_i} \dots \succ_{b_i} n^{b_i}_k$.
    \item For any $b_i' \in B_0'$ we create the preference list $ c_i' \succ_{b_i'} c_{i+1}'$.
    \item For any $b_i'' \in B_0''$ we create the preference list only consisting of $c_{i+1}''$.
    \item For any $c_i \in C_0$ we create the preference list  $n^{c_i}_1 \succ_{c_i} \dots \succ_{c_i} n^{c_i}_k$.
    \item For any $c_i' \in C_0'$ we create the preference list $a_{i-1}' \succ_{c_i'} a_{i}$.
    \item For any $c_i'' \in C_0''$ we create the preference list $a_{i-1}'' \succ_{c_i''} a_i''$ if $i$ is odd and $a_{i}'' \succ_{c_i''} a_{i-1}''$ if $i$ is even.
\end{itemize}
We now show that the matching $M = \lbrace (a_i, b_i', c_i') \mid a_i \in A_0\rbrace \cup \lbrace (a_i'', b_i, c_i'') \mid b_i \in B_0\rbrace $ is popular if and only if no perfect matching exists in~$\mathcal{I}$.

\subsubsection{$\Rightarrow$}
First assume that there is no perfect matching and let $\Mp$ be a matching different from $M$.
We distinguish two cases. 

In the first case we assume that every $b_i''$ is matched to $c_{i+1}''$.
For any $b_i$ there are now two cases. 
\begin{itemize}
    \item For any $b_i$ that is matched to an agent in $C_0$, we let $a_{j} = \Mp(\Mp(b_i))$ and $X_i = \{b_i, b_i'', c_{i+1}'', a_i'', a_j, a_j', b_j', c_j', \Mp(b_i)\}$. 
    \item For any $b_i$ with $\Mp(b_i) = b_i$ we can get a unique bijection to an agent $a_{j}$ with $\Mp(a_j) \in \{a_j, b_j'\}$, \ie agents from $A_0$ whose matching partner is not from $B_0$. In this case we again set $X_i = \{b_i, b_i'', c_{i+1}'', a_i'', a_j, a_j', b_j', c_j'\}$. 
\end{itemize}

As in the proof of Theorem~\ref{thrm_strongpop_np} we define $\vote_i(\Mp) = \sum_{x\in X_i} \vote_x(\Mp, M)$ and see that $\vote(\Mp, M) = \sum_{i = 1}^n \vote_i(\Mp)$.

Now let $i\in [n]$ be odd and assume that $\Mp(b_i) \in C_0$. Then since $b_{i}''$ is matched to $c_{i+1}''$ and $a_i''$ must be matched to $b_i''$, we get that at most $b_i'', c_{i+1}'', a_j', c_j', \Mp(b_i)$ might prefer $\Mp$ to $M$, while $b_i, a_j, a_i'', b_j'$ prefer $M$ to $\Mp$, which leads to $\vote_i(\Mp) \le 1$.
Similarly if $i\in [n]$ is even and $\Mp(b_i) \in C_0$, we conclude that $\vote_i(\Mp) \le -1$.

For $i\in [n]$ is odd and $\Mp(b_i) = b_i$ we observe that since $a_i''$ must be matched to $b_i''$, agents $c_{i+1}''$ and $b_i''$ prefer $\Mp$ to $M$, while $b_i$ and $a_i''$ prefer $M$ to $\Mp$. By the same reasoning as in Theorem~\ref{thrm_np_pop_inc} we also get that the sum of votes among $a_j, a_j', b_j', c_j'$ is at most $0$ and thus $\vote_i(\Mp) \le 0$. This argumentation also extends to an even $i$ with $\vote_i(\Mp) \le -1$. Since no matching is perfect, there must be at least one $b_i$ with $\Mp(b_i) = b_i$ and we get $\sum_{i = 1}^n \vote_i(\Mp) \le -1 +\sum_{i = 1}^n (-1)^n \le 0$. Therefore in this case $\Mp$ is not more popular than~$M$.

We now arrived to the second case, in which there is some $b_i''$ who is not matched to $c_{i+1}''$. Let $i_{1}, \dots, i_{k}$ be the sorted list of the indices of such agents in~$B''$. Further for any $i_j$, we define $X_{i_j}$ analogously to $X_i$ in the previous case. 

Now for any $i_j$, if $b_{i_j}$ is matched to an agent in $C_0$, we set $a_l = \Mp(\Mp(b_{i_j}))$ and get that at most $a_l', b_l'$, and $\Mp(b_{i_j})$ could get better, while $b_{i_j}, a_l, b_l', a_i''$ would get worse, which leads to $\vote_i(\Mp) \le -1$.%
If $b_{i_j}$ is not matched to an agent in $C_0$, then there are two cases. 
\begin{itemize}
    \item If $\Mp(b_{i_j-1}'') = c_{i_j}''$, then $a_{i_j}''$ and $b_{i_j}$ must prefer $M$ to $\Mp$, while $b_{i_j}''$ and $c_{i_j+1}''$ are at most indifferent. Combined with the observation that the vote of the other agents cannot be positive, we get $\vote_i(\Mp) \le -2.$
    \item  However if $\Mp(b_{i_j-1}'') \neq c_{i_j}''$, then either nothing changes in the matching or the only agent who can still improve is $a_l'$. However if $a_l'$ improves, then $a_l$ and $b_l$ must get worse, which implies that $\vote_i(\Mp)\le -1$. 
\end{itemize}

This implies that  $\sum_{\ell = i_j}^{i_j+1} \vote_\ell(\Mp)$ is at most $0$ if $i_j + 1 = i_{j+1}$ or $\sum_{\ell = i_j}^{i_j+1} \vote_\ell(\Mp) \le \sum_{j=1}^k 1 + \vote_{i_j}(\Mp) \le 0$ if $i_j = 1 \neq i_{j+1}$.
 Thus we get that $\vote(\Mp, M) = \sum_{j=1}^k \sum_{\ell = i_j}^{i_j+1} \vote_\ell(\Mp) \le 0$ and therefore $M$ is popular. 

\subsubsection{$\Leftarrow$}
If there is a perfect matching $M_0$ we can take the matching $\Mp = M_0 \cup \{(a_i', b_i, c_{i+1}'), (a_i'', b_i'', c_{i+1}'') \mid i \in [n]\}$ and observe that every $a_i, b_i, b_i', a_i'',$ and $c_i''$ with an even $i$ prefer $M$ to $\Mp$, while every $c_i, c_i', b_i'', a_i'$ and $c_i''$ with an odd $i$ prefer $\Mp$ to~$M$. Since $n$ is odd, this implies that $\vote(\Mp, M) \ge 1$ and therefore $\Mp$ is more popular than $M$.
\qed \end{proof}
\abpop*
\begin{proof}
	~\subsubsection{Construction} The construction of the proof is very similar to the \NP-hardness proof of \citet[Theorem 1]{BM10}. 
	We  reduce from the problem of finding a popular matching in bipartite instances with ties.
	More formally given a bipartite graph $G = (U \cup W, E)$ such that each agent in $U$ 
	has a strict preference list over a subset of agents in $W$ and each agent in $W$ either has a strict preference list or a single tie. 
	This problem was proven to be \NP-hard by \citet{CHK17}. 
	Let $W^t$ be the set of agents in $W$ who have a tie in their preference list. 
	We set
	\begin{itemize}
	    \item $A = \lbrace a_i \mid u_i \in U \rbrace$ to be a copy of $U$;
	    \item $B = \lbrace b_i \mid w_i \in W \rbrace$ to be a copy of $W$;
	    \item $C = \lbrace c_i \mid w_i \in W^t \rbrace \cup \lbrace c_{ij} \mid w_i \in W \setminus W^t, u_j \in U\rbrace$.
	\end{itemize} 
	Now we turn to the preferences.
	\begin{itemize}
	    \item Each $a_i \in A$ has the same preference list as the corresponding $u_i \in U$.
	    \item If $w_i \in W^t$, then $b_i$'s preference list is $c_i$ alone, while $c_i$'s preference list is the preference list of $w_i$. 
	    \item Further if $w_i \notin W^t$, we add the agents $c_{ij}$ in the order of the respective $u_j$ as the preference list of $b_i$ and add $a_j$ as the single element of $c_{ij}$'s list. 
	\end{itemize}
		
	\subsubsection{Correctness} First we notice that there is a one-to-one correspondence of matchings in our original graph and in our constructed \textsc{3dpmi} instance. 
	For a matching $M$ in $G$, we can construct the matching $\overline M \coloneqq \lbrace (a_j, b_i, c_{ij}) \mid (m_j, w_i ) \in M, w_i \notin W^t\rbrace \cup \lbrace (a_j, b_i, c_i) \mid (m_j, w_i ) \in M, w_i \in W^t\rbrace $ and vice versa. 
	We will show that $M$ is popular if and only if $\overline M$ is $A \cup B$-popular. 
	
	First assume that $M$ is not popular. Then there is a more popular matching~$\Mp$. If $m_i$ prefers $\Mp$ to $M$, then $a_i$ also prefers $\overline \Mp$ to $\overline M$ since the preference lists are the same. An analogous statement holds for $w_i \notin W^t$. If $w_i \in W^t$ and $w_i$ is matched in both $M$ and $\Mp$ then $w_i$ will also be indifferent between $\overline M$ and $\overline \Mp$, and if $w_i$ is matched in $\Mp$ but not in $M$ then $w_i$ will also prefer $\Mp$ to~$M$. Thus $\overline \Mp$ is also more $A\cup B$-popular than~$\overline M$.
	
	The same argument also holds for a matching in the $A \cup B$-popular instance that is not popular, and thus $M$ is popular if and only if $\overline M$ is $A \cup B$-popular. Consequently our instance has a popular matching if and only if the instance we reduced to has a $A \cup B$-popular matching. 
\qed \end{proof}

\end{document}